\documentclass[copyright,creativecommons]{eptcs}
\providecommand{\event}{ICE 2016} 

\usepackage[english]{babel}
\usepackage[utf8x]{inputenc}
\usepackage{hyperref}
\usepackage{url}
\usepackage{graphicx}
\usepackage{xcolor}
\usepackage{amsmath,amssymb,amsfonts,amsthm,thmtools,thm-restate}
\usepackage{latexsym}

\usepackage{mymacros}

\newtheorem{fact}{Fact}[section]

\newtheorem{theorem}[fact]{Theorem}

\newtheorem{definition}[fact]{Definition}

\title{A Step Towards Checking Security in IoT 
\thanks{Partially supported by Universit\`a di Pisa PRA\_2016\_64 Project \emph{Through the fog}.}} 

\author{Chiara Bodei 
\qquad
Pierpaolo Degano
\qquad
Gian-Luigi Ferrari
\qquad
Letterio Galletta
\institute{Dipartimento di Informatica, 
Universit\`a di Pisa}
\email{ \{chiara,degano,giangi,galletta\}@di.unipi.it}
}
\def\titlerunning{A Step Towards Checking Security in IoT}
\def\authorrunning{C.~Bodei, P.~Degano, G-L.~Ferrari, L.~Galletta}

\begin{document}
\maketitle

\begin{abstract}
The Internet of Things (IoT) is smartifying our everyday life.
Our starting point is \IoTLySa, a calculus for describing IoT systems, and its static analysis, which will be presented at Coordination 2016.
We extend the mentioned proposal in order to begin an investigation about security issues, in particular for the static verification of secrecy and some other security properties. 
\end{abstract}

\section{Introduction}\label{sec:intro}
The Internet of Things (IoT) is pervading our everyday life. 
As a consequence, it is crucial to formally reason about IoT systems, to understand and govern
the emerging technology shifts.
In a companion paper~\cite{BDFG_Coord16}, we 
introduced \IoTLySa, a dialect of {\LySa}~\cite{BBDNN_JCS,BNN03}, within the process calculi approach to IoT~\cite{lanese13,merro16}.
It has primitive constructs to describe the activity of sensors and of actuators, and to manage the coordination and communication of smart objects.
More precisely, our calculus consists of:
\begin{enumerate}
\item systems of nodes, made of physical components, i.e.~sensors and actuators, and of SW control processes for specifying the \emph{logic} of the node, e.g.~the manipulation of sensor data;

\item a shared store \`a la Linda~\cite{G85,CG01} to implement intra-node generative communications.  
The adoption of this coordination model supports a smooth implementation of the 
\emph{cyber-physical control architecture}:
physical data are made available to software entities that analyse them and trigger the relevant actuators to perform the desired behaviour;

\item an asynchronous multi-party communication among nodes, which can be easily tuned to take care of various constraints, e.g.~those concerning proximity or security;

\item functions to process and aggregate data.
\end{enumerate}

\noindent
Moreover, we equipped \IoTLySa\ with a Control Flow Analysis (CFA) that safely approximates system behaviour: it describes the interactions among nodes, how data spread from sensors to the network, and how data are manipulated.
%
Technically, our CFA abstracts from the concrete values and only considers  their provenance and how they are put together. 
In more detail, for each node $\ell$ in the network it returns: 
\begin{itemize}
\item
an abstract store $\hat{\Sigma}_\ell$ that records for each variable the set of the abstract values that it may denote; 
\item
a set $\kappa(\ell)$ that over-approximates the set of the messages received by the node $\ell$;
\item
a set $\Theta(\ell)$ of possible abstract values computed and used by the node $\ell$.
\end{itemize}
%
In this paper, we show that our CFA can be used as the basis for checking  some interesting properties of IoT systems.
An important point is that our verification techniques take as input the result of the analysis, with no need of recomputing them.
We extend the analysis in two main directions: the first tracks how actuators may be used and the second addresses security issues.
We think that the last point is crucial, because IoT systems handle a huge amount of sensitive data (personal data, business data, etc.), and because of their capabilities to affect the physical environment~\cite{IERC}.

Technically, we track the behaviour of actuators through a further component of our analysis, $\alpha_{\ell}(j)$ that contains all the possible actions the $j^{th}$ actuator may perform in the node $\ell$. 
In this way one can estimate the usage of an actuator so helping the overall design of the system, e.g.\ taking also care of access control policies.

To deal with security here we include encryption and decryption, not present in~\cite{BDFG_Coord16}.
Since cryptographic primitives are often power consuming and sensors and actuators have little battery, our extended CFA may help IoT designers  in 
detecting which security assets are to protect inside a node and in addition where cryptography is necessary or redundant. 
%
%
%
%
Indeed by classifying (abstract) values as secret or public, we can check whether values that are classified in a node are never sent to other untrusted nodes.
More generally, sensors and nodes can be assigned specific security clearance or trust levels.
Then we can detect if nodes with a lower level can access data of entities with a higher level by inspecting the analysis results.
Similarly, we can check if the predicted communication flows are allowed by a specific policy. 
In this way, one can evaluate the security level of the system in hand and discover potential vulnerabilities.

The paper is organised as follows.
In the next section, we intuitively introduce our approach by recalling the same illustrative example detailed in~\cite{BDFG_Coord16}.
In Sect.~\ref{sec:semantics} 
we briefly present the process calculus \IoTLySa, and we define our CFA in Sect.~\ref{sec:analysis}.
In Sect.~\ref{sec:extendedCFA} we adapt and apply our CFA in order to capture some security properties.
Concluding remarks and related work are in Sect.~\ref{sec:conclusion}. 

\section{A smart street light control system}\label{sec:example}


In this section, we briefly recall the smart street light control system modelled in \IoTLySa\ in~\cite{BDFG_Coord16}.
These systems represent effective solutions to improve energy efficiency~\cite{Elejoste,ECMC14} that is relevant for IoT~\cite{IERC}.

We consider a simplified system made of two integrated parts, working on a one-way street.
The first consists of smart lamp posts {$N_p$} that are battery powered and can sense their surrounding environment and can communicate with their neighbours to share their views. 
If (a sensor {$S_{p,i}$} of) the lamp post perceives a pedestrian and there is not enough light in the street it switches on the light and communicates the presence of the pedestrian to the lamp posts nearby.
When a lamp post detects that the level of battery is low, it informs the supervisor $N_s$ of the street lights that will activate other lamp posts nearby.
The second component {$N_{cp}$} of the street light controller uses the electronic access point to the street. 
When a car crosses the checkpoint, besides detecting if it is enabled to, a message is sent to the supervisor of the street accesses, $N_a$, that in turn notifies the presence of the car to $N_s$.
This supervisor sends a message to the lamp post $N_p$ closest to the checkpoint that starts a forward chain till the end of the street.
The structure of our control light system is in \figurename~\ref{fig:example}, while its \IoTLySa\ specification $N$ is in \tablename~\ref{SmartStreetLight}.
The whole intelligent controller $N$ of the street lights is described as the parallel composition of the checkpoint node $N_{cp}$, the supervisors nodes $N_a$ and $N_s$, and the nodes of lamp posts $N_p$, with $p \in [1,k]$.

\begin{figure}[t]
\centering
\includegraphics[scale=0.45]{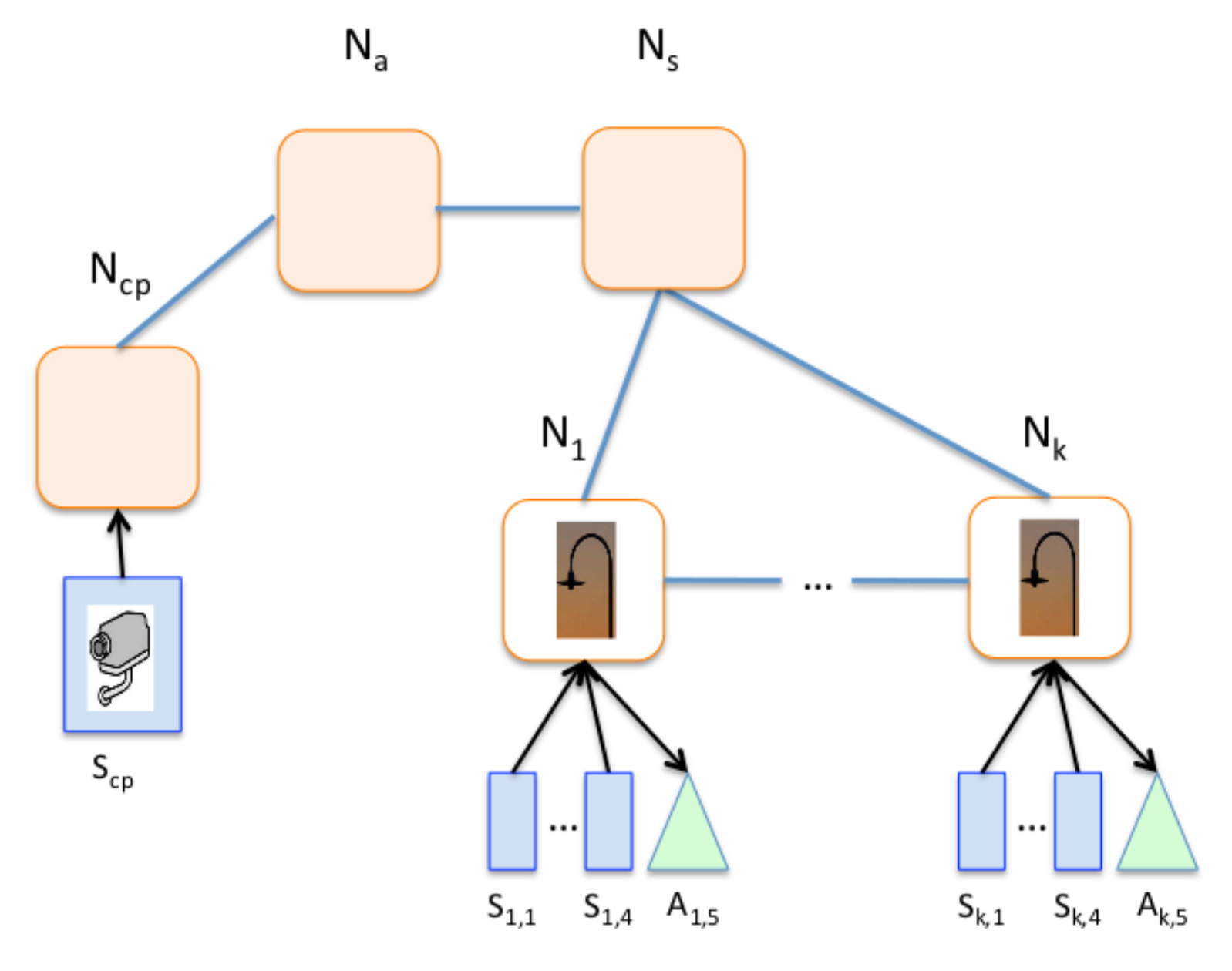}
\caption{The organisation of nodes in our street light control system.}
\label{fig:example}
\vspace{-.3cm}
\end{figure}

{\small
\begin{table*}[!tb]

 \[
 \begin{array}{l}
\mbox{\bf Checkpoint } 
\
N_{cp} = \ell_{cp} : [ P_{cp} \ \|\ S_{cp}  \ \|\ B_{cp} ]
\\
P_{cp} = \mu h. (z := 1).(z' := noiseRed(z)).\OUTM{z'}{\{\ell_a\}}. \ h
\\
S_{cp} = \mu h. (\tau.1 := v_{p}). \tau. \ h
\\
\mbox{\bf Street Supervisor} \
N_{a} = \ell_a : [\mu h. \INPS{}{x}{\OUTM{car, x}{\{\ell_s\}.\ h}} \ \|\ B_{a} ]
\\
\mbox{\bf Lamp Supervisor}  
\
N_s = \ell_s : [\mu h.\, \INPS{err}{x}{\OUTM{true}{L_x}}. \ h\, \ \|\ P_{s,1}  \ \|\ B_s]
\\
P_{s,1} = 	\mu h. \INPS{car}{x}{\OUTM{x}{\{\ell_{1}\}}. \ h}
\\
\mbox{\bf Lamp Post } \mbox{ with }p \in [1,k]
\
N_p = \ell_p : [\Sigma_p \ \|\ P_{p,1} \ \|\ P_{p,2} \ \|\ S_{p,1} \ \|\ S_{p,2} \ \|\ S_{p,3} \ \|\ S_{p,4} \ \|\ A_{p,5}]
\\
\!\!\!\! \begin{array} {ll}
P_{p,1} = \mu h. & (x_1:= 1.\, x_2:= 2.\, x_3:= 3.\, x_4:= 4). \\
& (x_4 = true)\ ?\  \\
& \qquad (x_1 \leq th_1 \land x_2 \leq th_2)\ ?\ \\
 &  \qquad \qquad \qquad \qquad\ \  \ (x_3 > th_3)\ ?\ \OUTS{5,\mathsf{turnon}}{\OUTM{x_4}{L_p}}.\ h\\
 &  \qquad \qquad \qquad \qquad \qquad \qquad \quad \ \ : \ \OUTM{\mathsf{err}, \ell_p}{\{\ell_s\}}. \ h \\
  &  \qquad \qquad \qquad \qquad \qquad \qquad : h \\
    & \ \ \ \ \qquad \qquad: \  \OUTS{5,\mathsf{turnoff}}{h}
\end{array}
\\
P_{p,2} = \mu h. \INPS{}{x}{(x = true \lor is\_a\_car(x)) \ ? \ (\OUTS{5,\mathsf{turnon}}{\OUTM{x}{L_p}). h}:  \OUTS{5,\mathsf{turnoff}} h}
\\
S_{p,i} = \mu\,h.\, (i := v). \ \tau.\;h
\\
A_{5} = \mu\,h.\, (\!|5, \{\mathsf{turnon}, \mathsf{turnoff} \}|\!).\;h
\\
\end{array} 
\] 
\caption{Smart Street Light Control System $N = N_{cp} \mid N_a \mid N_s \mid N_1 \mid \dots \mid N_k$.}\label{SmartStreetLight}
\end{table*}
}

The checkpoint $N_{cp}$ is an \IoTLySa\ node with label $\ell_{cp}$ that only contains a visual sensor $S_{cp}$ to take a picture of the car detected in the street, 
a process $P_{cp}$ and the component $B_{cp}$, which abstracts other components we are not interested in.
The sensor communicates the picture $v_p$ to the node by storing it in the location $1$ of the shared store ($1$ is the identifier of the sensor $S_{cp}$ in the assignment $z := 1$). 
In our model we assume that each sensor has a reserved store location in which records its readings.
The action $\tau$ denotes internal actions of the sensor, which we are not interested in; the construct $\mu$ implements the iterative behaviour of the sensor.
Then, the taken picture is enhanced by the process $P_{cp}$, by using the function $noiseRed$ to reduce the noise and sent to the supervisor $N_a$.

The node $N_a$ receives the enhanced picture from $N_{cp}$, (checks if the car is allowed to enter the street and) and communicates its presence to the lamp posts supervisor $N_s$.

The process $P_{s,1}$, inside the supervisor $N_s$, receives the picture from $N_a$ and sends a message to the node closest to the checkpoint $N_1$, 
with label $\ell_1$.
The input $(car; x)$ is performed only if the corresponding output matches the constant $car$, and the store variable $x$ is bound to the value of the second element of the output (see the full definition of $N_s$).
 
In our intelligent street light control system there is a node $N_p$ for each lamp post, each of which has a unique identifier $p \in [1,k]$.
Each lamp post $N_p$ has a store $\Sigma_p$ shared with its components and 
four sensors to sense $(1)$ the environment light, $(2)$ the solar light, $(3)$ the battery level and $(4)$ the presence of a pedestrian. 
Each of them is defined by $S_{p,i}$
where $v$ is the perceived value and $i \in [1,4]$ are the store locations for the sensors.  
After some internal actions $\tau$, the sensor $S_{p,i}$ iterates its behaviour.
The actuator for the lamp post $p$ is defined by $A_{5}$.
It only accepts a message from $N_c$ whose first element is its identifier (here 5) and whose second element is either command $\mathsf{turnon}$ or $\mathsf{turnoff}$ and executes it. 

Two parallel processes compose the control process $N_{p}$ of a lamp post node: $P_{p,1}$ and  $P_{p,2}$. 
The first one reads the current values from the sensors and stores them into the local variables $x_i$. 
The actuator is turned on if (i) a pedestrian is detected in the street ($x_4$ holds), (ii) the intensity of environment and solar lights are greater than the given thresholds 
$th_1$ and $th_2$, and (iii) there is enough battery (more than $th_3$).
In addition, the presence of the pedestrian is communicated to the lamp posts nearby, whose labels, typically $\ell_{p-1}$ and $\ell_{p+1}$, are in $L_p$.
Instead, if the battery level is insufficient, an error message, including its identifier $\ell_p$, is sent to the supervisor node, labelled $\ell_s$.
%
The second process waits for messages from its neighbours or from the supervisor node $N_s$.
When one of them is notified the presence of a pedestrian ($x = true$) or of a car ($is\_a\_car(x)$ holds), the current lamp post orders the actuator to switch the light on.
%

In the lamp post supervisor $N_s$, in parallel with $P_{s,1}$, there is the process $\mu h.\, \INPS{err}{x}{\OUTM{true}{L_x}}. \ h$.
As above the input $(err; x)$ is performed only if the corresponding output matches the constant $err$, and the store variable $x$ is bound to the value of the second element of the output i.e.\ the label of the relevant lamp post.
In this case, 
$N_s$ warns the lamp posts nearby $x$ (included in $L_x$) of the presence of a pedestrian.

We would like to statically predict how the system behaves at run time. 
In particular, we want to compute:
$(i)$ how nodes interact with each other; 
$(ii)$ how data spread from sensors to the network (tracking); 
$(iii)$ which computations each node performs on the received data; and (iv) which actions an actuator may trigger.
To do that, we define a Control Flow Analysis (CFA), which abstracts from the concrete values by only considering their provenance and how they are manipulated. 
Consider the picture sent by the camera of $S_{cp}$ to its control process $P_{cp}$. 
In the CFA we are only interested in tracking where the picture comes from, and not in its actual value; 
so we use the  abstract value $1^{\ell_{cp}}$ to record the camera that took it. 
The process $P_{cp}$ reduces the noise in the pictures and sends the result to $N_a$. 
Our analysis keeps track of this manipulation through the abstract value $noiseRed^{^{\ell_{cp}}}(1^{\ell_{cp}})$, 
meaning that the function $noiseRed$, in the node $\ell_{cp}$, is applied to data coming from the sensor with identifier $1$ of $\ell_{cp}$.

In more detail, our CFA returns for each node $\ell$ in the network: 
an abstract store $\hat{\Sigma}_{\ell}$ that records for each variable a super-set of the abstract values that it may denote at run time; 
a set $\kappa(\ell)$ that approximates the set of the messages received by the node $\ell$; 
the set $\Theta(\ell)$ of possible abstract values computed and used by the node $\ell$; and the set $\alpha_{\ell}(j)$ that collects all the actions that the actuator $j$ can trigger.

Here, for each lamp post labelled $\ell_p$ and for its actuator $j$ the analysis returns $(i)$ in $\alpha_{\ell_p}(j)$ the actions $\mathsf{turnon},\mathsf{turnoff}$;
$(ii)$ in $\kappa(\ell_p)$ both the abstract value $noiseRed^{^{\ell_{cp}}}(1^{\ell_{cp}})$ and the sender $\ell_{p+1}$ of that message.
{We can exploit the result of our analysis to perform some security verifications (see Sect.~\ref{sec:extendedCFA}), 
e.g.~since the pictures of cars are sensitive data, one would like to check whether they are kept secret.
By inspecting $\kappa$
we discover that the sensitive data of cars is sent to all lamp posts, so possibly violating privacy.
Also, we could verify if the predicted communications between nodes respect a given policy, e.g.\ the trivial one prescribing that the interactions between the nodes $N_s$ and $N_a$ are strictly unidirectional.

%
}

\section{The calculus \IoTLySa}\label{sec:semantics}

We model IoT applications through the process calculus \IoTLySa~\cite{BDFG_Coord16} that consists of 
(i) systems of nodes, in turn consisting of sensors, actuators and control processes, plus a shared store $\Sigma$ within each node for internal communications;
(ii) primitives for reading from sensors, and for triggering actuator actions;
(iii) an asynchronous multi-party communication modality among nodes, subject to constraints, e.g.~concerning
proximity or security;
(iv) functions to process data;
(v) explicit conditional statements.
Differently from~\cite{BDFG_Coord16}, we include here encryption and decryption primitives as in {\LySa}, taking a symmetric schema.
We assume as given a finite set $\mathcal K$ of secret keys owned by nodes, previously exchanged in a secure way, as it is often the case~\cite{ZigBee}.
Since our analysis at the moment considers  no malicious participants, this is not a limitation; we can anyway implicitly include attackers, following~\cite{BBDNN_JCS}.

Systems in \IoTLySa\ have a two-level structure and consist of a fixed number of uniquely labelled nodes, hosting a store, control processes, sensors and actuators.
The syntax follows:
\[
\begin{array}{ll@{\hspace{2ex}}l}
{\cal N} \ni N ::= & {\it systems \ of \ nodes} &\\
& \NIL                       & \hbox{inactive node} \\
& \ell: [B]                      & \hbox{single node}\;  (\ell \in {\cal L} \text{, the set of labels})
\\
& N_1\ |\ N_2 &  \hbox{parallel composition of nodes}
\\
{\cal B} \ni B ::= & \text {\it node components} &\\
&  \Sigma_\ell  & \hbox{node store}
\\
& P  & \hbox{process}
\\
& S   & \hbox{sensor, with a unique identifier $i \in \cal{I}_{\ell}$}
\\
& A     & \hbox{actuator, with a unique identifier $j \in \cal{J}_{\ell}$}
\\
& B \ \|\ B & \hbox{parallel composition of node components}
\end{array}
\]
%
The label $\ell$ uniquely identifies the node  
$\ell: [B]$ and may represent further characterising information (e.g.\ its location or other contextual information, see Sect.~\ref{sec:extendedCFA}).
Finally, the operator $|$  describes a system of nodes obtained by parallel composition.
 
We impose that in $\ell: [B]$  there is a \emph{single} store $ \Sigma_{\ell} : \cal{X} \cup \cal{I}_{\ell} \ \rightarrow \cal{V} $, where ${\cal X}, {\cal V}$  are the sets of variables and of values (integers, booleans, ...), resp.
Our store is essentially an array of fixed dimension, so intuitively a variable is the index in the array and an index $i \in \cal{I}_{\ell}$ corresponds to a single sensor (no need of $\alpha$-conversions).
We assume that store accesses are atomic, e.g.\ through CAS instructions~\cite{CAS91}.
The other node components are obtained by the parallel composition of control processes $P$, and of a fixed number of (less than 
$\#(\cal{I}_{\ell})$) sensors $S$, and actuators $A$ (less than 
$\#(\cal{J}_{\ell})$) the actions of which are in $Act$.
The syntax of processes is as follows:
\[
\begin{array}{lll}
{\cal P} \ni P ::= & {\it control \ processes} &\\
& \NIL                       & \hbox{inactive process} \\
& \OUTM{E_1, \cdots, E_k}{L}.\,P & \hbox{asynchronous multi-output L$\subseteq {\cal L}$} \\
& \INPS{E_1,\cdots,E_j}{x_{j+1},\cdots,x_k}{P}\  & \hbox{input (with matching)}\\
& \DECSO{E}{E_1,\cdots,E_j}
      {x_{j+1},\cdots,x_k}{k_0}{}{P} & \hbox{decryption with key $k_0$ (with match.)}
      \\
& E?P:Q &  \hbox{conditional statement} \\
& h   &  \hbox{iteration variable}
\\
&\mu h. \ P & \hbox{tail iteration}
  
\\[.2ex]
& x := E.\,P & \hbox{assignment to $x \in {\cal X}$}
\\
& \OUTS{j, \gamma}{P}& \hbox{output of action $\gamma$ to actuator $j$}
\end{array}
\]

\noindent
The prefix $\OUTM{E_1, \cdots, E_k }{L}$ implements a simple form of multi-party communication among nodes: the tuple $E_1, \dots, E_k$ is asynchronously sent to the nodes with labels in $L$ and that are ``compatible'' (according, among other attributes, to a proximity-based notion).
The input prefix $(E_1, \!\cdots\!,E_j; x_{j+1},\! \cdots \!,x_k)$
is willing to receive a $k$-tuple, provided that its first $j$ elements match the corresponding input ones, and then binds the remaining store variables (separated by a ``;'') to the corresponding values (see \cite{BNN03,BBF15} for a more flexible choice).
Otherwise, the $k$-tuple is not accepted.
A process repeats its behaviour, when defined through the tail iteration construct $\mu h. P$ (where $h$ is the iteration variable).
The process $\DECSO{E}{E_1,\cdots,E_j}{x_{j+1},\cdots,x_k}{k_0}{}{P}$
receives a message encrypted with the shared key $k_0 \in \mathcal K\!$.
Also in this case we use the pattern matching but additionally 
the message $E =\{E'_1, \cdots, E'_k\}_{k_0}$ is decrypted with the key $k_0$. 
Hence, whenever $E_i = E'_i$ for all $i\in [0,j]$, the receiving process behaves as $P\{E_{j+1}/x_{j+1},\ldots,E_{k}/x_{k}\}$.
Sensors and actuators have the form:
\[
\begin{array}{ll@{\hspace{2ex}}l ll@{\hspace{2ex}}l}
{\cal S} \ni S &::=  {\it sensors} & &
{\cal A} \ni A &::= {\it actuators} &\\
& \NIL &  \hbox{\hspace{-5mm}inactive sensor} &&
 \NIL &  \hbox{\hspace{-5mm}inactive actuator}
\\
& \tau.S & \hbox{\hspace{-5mm}internal action} &
& \tau.A & \hbox{\hspace{-5mm}internal action} 
\\
& i :=  v.\,{S} & \hbox{\hspace{-5mm}store of $v \in {\cal V}$} 
&& 
(\!|j, \Gamma|\!).\,A & \hbox{\hspace{-5mm}command for actuator $j$ ($\Gamma \subseteq Act$)} 
\\
&& \hbox{\hspace{-5mm}by the $i^{th}$ sensor}
 &&
\gamma.A & \hbox{\hspace{-5mm}triggered action ($\gamma \in Act$)}

\\
& h & \hbox{\hspace{-5mm}iteration var.} 
&&
h & \hbox{\hspace{-5mm}iteration var.} 
\\
& \mu \,h\,.\, S & \hbox{\hspace{-5mm}tail iteration}
&& \mu \,h\,.\, S & \hbox{\hspace{-5mm}tail iteration}
\end{array}
\]

\noindent
We recall that sensors and actuators are identified by 
unique identifiers.
A sensor can perform an internal action $\tau$ or store the value $v$, gathered from the environment, into its store location $i$.
An actuator can perform an internal action $\tau$ or execute one of its action $\gamma$,
possibly received from its controlling process.
Both sensors and actuators can iterate.
For simplicity, here we neither provide an explicit operation to read data from the environment, nor to describe the impact of actuator actions on the environment.
Finally, the syntax of terms is as follows:
\[
\begin{array}{ll@{\hspace{2ex}}l}
{\cal E} \ni E ::= & 
       {\it terms}&\\
& v  & \hbox{value } (v \in {\cal V})\\
& i  & \hbox{sensor location } (i \in {\cal I_\ell})\\
& x  & \hbox{variable } (x \in {\cal X})\\
& \{E_1,\cdots,E_k\}_{k_0}
              & \hbox{encryption with key }k_0\; (k\geq 0)
              \\
& f(E_1, \cdots, E_n) &  \hbox{function on data} \\
\end{array}
\]
The encryption function $\{E_1,\cdots,E_k\}_{k_0}$ returns the result of encrypting values $E_i$ for $i \in [1,k]$ under $k_0$ representing a shared key in $\mathcal K\!$.
The term $f(E_1, \cdots, E_n)$ is the application of function $f$ to $n$ arguments; 
we assume given a set of primitive functions, typically for aggregating or comparing values, be them computed or representing data in the environment.
We assume the sets $\mathcal V\!$, $\mathcal I\!$, $\mathcal J\!$, $\mathcal K\!$ be pairwise disjoint.



The semantics is based on a \emph{structural congruence} $\equiv$ on nodes, processes, sensors, and actuators.
It is standard except for the last rule that equates a multi-output with no receivers and the inactive process, and for the fact that inactive components of a node are all coalesced.

\[
\begin{array}{ll}
-\  & ({\mathcal N}/_{\equiv}, \mid, \NIL)  \text{ and } ({\mathcal B}/_{\equiv}, \|, \NIL)  \mbox{ are commutative monoids} 
\\
-\ &\mu \,h\,.\, X \equiv X\{\mu \,h\,.\, X/h\}  \quad
 \text{ for } X \in \{P, A, S\}
\\
-\  & \langle \langle E_1,\cdots,E_k \rangle \rangle : \emptyset . \ \NIL\equiv \NIL
\end{array}
\]

We have a two-level {\em reduction relation\/} $\rightarrow$ defined as the least relation on nodes and its components satisfying the set of inference rules in \tablename~\ref{opsem}.
We assume the standard denotational interpretation $\dsem{E}_\Sigma$ for evaluating terms.

\begin{table*}[htb]
{\small
\[
\begin{array}{l}
\begin{array}{ll}
\begin{array}{l}
  \mbox{(S-store)}\\      
  \irule{}{
                    \Sigma \ \| \ i := v.\,{S_i \ \|\ B
               \rightarrow
               \Sigma\{v/{i}\} \ \| \ S_i \ \| \ B}
              }
              \hfill
\end{array} 
\quad & \quad
\begin{array}{l}
 \mbox{(Asgm)}
 \\       \irule{ \dsem{E}_\Sigma = v}{
                     \Sigma \  \|\  x:=E.\,P \  \|\ B \ 
               {\rightarrow} \ 
              \Sigma\{v/x\}  \ \| \   P \ \| \ B 
              }
              \end{array}
\hfill
\end{array}
\\[6ex]

\begin{array}{l}
\mbox{(Ev-out)} \\ 
\irule{ \bigwedge_{i=1}^{k} {v_i}= \dsem{E_i}_{\Sigma} 
}{
\Sigma  \  \|\   \OUTM{E_1,\cdots,E_k}{L}.\,P \  \|\  B \ \rightarrow \ 
   \Sigma  \  \|\   \OUTM{v_1,\cdots, v_k}{L} . \NIL\  \|\  P \  \|\  B 
}
\end{array}
\hfill
\\[7ex]

\mbox{(Multi-com)}\\ 
\irule{\ell_2 \in L \wedge \ Comp(\ell_1,\ell_2)  \wedge \  \bigwedge_{i=1}^{j} {v_i}= \dsem{E_i}_{\Sigma_2} }
             {
             \begin{array}{c}
             \ell_1: [\OUTM{v_1,\cdots,v_k}{L}. \, \NIL \  \|\  B_1] \ \ | \ \
             \ell_2:[\Sigma_2 \ \| \ (E_1,\cdots,E_j;x_{j+1},\cdots,x_k).Q \ \| \ B_2] \ \rightarrow
              \\[.3ex]
              \ell_1: [\OUTM{v_1,\cdots,v_k}{L \setminus \{\ell_2\}}. \, \NIL \ \| \ B_1]
            \  \ | \ \ \ell_2: [\Sigma_2\{v_{j+1}/x_{j+1},\cdots,v_k/x_k\} \ \| \ Q\ \| \ B_2]
              \end{array}}
\\[8ex]
 \mbox{(Decr)} \\       \irule{
 \dsem{E}_{\Sigma} = \{v_1,\cdots,v_k\}_{k_0} \ \wedge \
  \bigwedge_{i=1}^{j} {v_i}= \dsem{E'_i}_{\Sigma}
\hspace{3ex}
}{
{\Sigma \ \| \ \DECSO{E}{E'_1,\cdots,E'_j}{x_{j+1},\cdots,x_k}{k_0}{}{P}} \| \ B
\ARROW \Sigma \| \  P[v_{j+1}/x_{j+1},\cdots,v_k/x_k] \ \| \ B \hfill
}
              \hfill
\\[5ex]
\begin{array}{ll}
\begin{array}{l}

 \mbox{(Cond1)} \\       \irule{ \dsem{E}_\Sigma = \tt{true}} {
                     \Sigma \  \|\  E?\,P_1:P_2 \  \|\ B \ 
              \rightarrow \ 
              \Sigma  \ \| \   P_1 \ \| \ B 
              }
              \end{array}
       &
       \begin{array}{l}

   \mbox{(Cond2)} \\       \irule{ \dsem{E}_\Sigma = \tt{false}} {
                     \Sigma \  \|\  E?\,P_1:P_2 \  \|\ B \ 
              \rightarrow \ 
              \Sigma  \ \| \   P_2 \ \| \ B 
              }
\end{array}
 \end{array}           
\\[5ex]
\begin{array}{ccc}
\begin{array}{l}
							\mbox{(A-com)}\\ 
\irule{\gamma \in \Gamma}{
                    \OUTS{j,\gamma}{P \  \|\   (\!|j,\Gamma|\!).\,A \ \| \ B \
             }
             \rightarrow
              \ P \ \| \ \gamma.\,A\  \|\ B 
              }
              \end{array}
\quad & \quad
 \begin{array}{l}
 \mbox{(Act)} \\ 
\irule{
}{
\gamma.A \ \rightarrow\  A
}
\end{array}
 \quad & \quad
 \begin{array}{l}
 \mbox{(Int)} \\   
    \irule{}{ \tau.\,X \ \rightarrow\  X} \ 
    \end{array}
\end{array}
              \hfill
\\[8ex]


\begin{array}{llll}
\begin{array}{l}
\mbox{(Node)}\\ 
\irule{B \ \rightarrow \ B'}
             {\ell: [B] \ \rightarrow \ \ell:[B']
             }
             \end{array}
\   &  \ 
\begin{array}{l}
\mbox{(ParN)}\\ 
       \irule{N_1 \rightarrow N'_1}
             {{N_1| N_2} \rightarrow {N'_1| N_2}}
			 \end{array}\ 
			 
\   &  \ 
\begin{array}{l}
\mbox{(ParB)}\\ 
       \irule{B_1 \rightarrow B'_1}
             {{B_1\| B_2} \rightarrow {B'_1\| B_2}}
			 \end{array}\ 
			 
\   &  \
\begin{array}{l}
\mbox{(CongrY)}\\ 
       \irule{Y_1' \equiv Y_1 \rightarrow Y_2 \equiv Y'_2}
             {Y'_1 \ARROW Y'_2}
			 \end{array}\ 
\end{array}
\hfill

\end{array}
\\ \\
\]}
\caption{Reduction semantics, where $X \in \{S, A\}$ and $Y \in \{N, B\}$.}
\label{opsem}
\end{table*}

%
The first two rules implement the (atomic) asynchronous update of shared variables inside nodes, by using
the standard notation $\Sigma\{-/-\}$.
According to (S-store), the $i^{th}$ sensor uploads the value $v$, gathered from the environment, into the store location $i$. 
According to (Asgm), a control process updates the variable $x$ with the value of $E$.
The rules (Ev-out) and (Multi-com) drive asynchronous multi-communications among nodes.
In the first 
a node labelled $\ell$ willing to send a tuple of values $\mess{v_1,...,v_k}$, obtained by the evaluation 
of $\mess{E_1,...,E_k}$,
spawns a new process, running in parallel with the continuation $P$;
its task is to offer the evaluated tuple to all its receivers $L$.
In the rule (Multi-com), the message coming from $\ell_1$ is received by a node labelled $\ell_2$. 
The communication succeeds, provided that: (i) $\ell_2$ belongs to the set $L$ of possible receivers, 
(ii) the two nodes are compatible according to the compatibility function $Comp$,
and (iii) that the first $j$ values match with the evaluations of the first $j$ terms in the input. 
Moreover, the label $\ell_2$ is removed by the set of receivers $L$ of the tuple. 
The spawned process terminates when all its receivers have received the message (see the last congruence rule).
The role of the compatibility function $Comp$ is crucial in modelling real world constraints on communication.
A basic requirement is that inter-node communications are proximity-based, i.e.~only nodes that are in the same transmission range can directly exchange messages.
Of course, this function could be enriched in order to consider finer notions of compatibility.
This is easily encoded here by defining a predicate (over node labels) yielding true only when two nodes respect the given constraints.
The inference rule (Decr) expresses the result of matching
the term
$ \{E_1,\cdots,E_k\}_{k_0}$,
resulting from an encryption, against the pattern in
$\mathsf{decrypt}\;{E}\;\mathsf{as}\; \{E'_1,\cdots,E'_j ;
x_{j+1},\cdots,x_k\}_{k_0}$
$\mathsf{in}\;{P}$, i.e.~the pattern occurring in the corresponding
decryption.  
As for communication, the value $v_i$ of each ${E_i}$
must match that of the corresponding ${E'_i}$ for the first $j$
components and in addition the keys must be the same (this models \emph{perfect}
symmetric cryptography). 
When successful, the values of the remaining expressions are bound to the corresponding variables.

According to the evaluation of the expression $E$, the rules for conditional are as expected.
A process commands the $j^{th}$ actuator $(\!|j,\Gamma|\!).\,A$ through the rule (A-com), by sending it the pair $\langle j, \gamma \rangle$; $\gamma$ prefixes the actuator, 
if it is one of its actions.
The rule (Act) says that the actuator performs the action $\gamma$. 
Similarly,
for the rules (Int) for internal actions for representing activities we are not interested in.
The last rules propagate reductions across parallel composition ((ParN) and (ParB)) and nodes (Node), while the
(CongrY) are the standard reduction rules for congruence.

\section{Control flow analysis}\label{sec:analysis}

We now revisit the CFA in~\cite{BDFG_Coord16} that only tracks the ingredients of the data handled by IoT nodes; more precisely it predicts where sensor values are gathered and how they are propagated inside the network be they raw data or resulting from computation.
Then, we extend it to also predict the actions of actuators, by including a new component $\alpha$, that for every actuator $j$ collects the actions $\gamma$ that may be triggered by the control process 
in the node labelled $\ell$.
Our CFA aims at safely approximating the abstract behaviour of a system of nodes $N$.  
We conjecture that the introduction of the component $\alpha$ does not change the low polynomial complexity of the analysis of~\cite{BDFG_Coord16}.

We resort to abstract values for sensor and  functions on abstract values, as follows, where $\ell \in \cal L$:
%
%
%
%
\[
\begin{array}{ll@{\hspace{2ex}}l}
\widehat{\cal V} \ni \hat{v} ::= & 
       {\it abstract\ terms \ }  &\\
& \top^{\ell}  & \hbox{special abstract value denoting cut} \\
& i^{\ell}  & \hbox{sensor abstract value } (i \in {\cal I}_\ell)\\
& v^{\ell}  & \hbox{node abstract value } \\
&\{\hat{v}_1, \cdots, \hat{v}_n\}^{\ell}_{k_0} &  \hbox{encryption on abstract data} \\
& f^\ell(\hat{v}_1, \cdots, \hat{v}_n) &  \hbox{function on abstract data} \\
\end{array}
\]
Since the dynamic semantics may introduce function terms with an arbitrarily nesting level, we have new special abstract values $\top^{\ell}$ that denote all those function terms with a depth greater that a given $d$.
In the CFA clauses, we will use $\lfloor - \rfloor_d$ to keep the maximal depth of abstract terms less or equal to $d$ (defined as expected).
Once given the set of functions $f$ occurring in $N$, the abstract values are finitely many.

{\small
\begin{table*}[!tb]

 \[
 \begin{array}{c}
 \irule{{i^{\ell}} \in \vartheta \subseteq  \Theta(\ell)}{\nEFORM{i}{\vartheta}}
\qquad \qquad
\irule{{v^{\ell}} \in \vartheta \subseteq  \Theta(\ell)}{\nEFORM{v}{\vartheta}}
\qquad \qquad
\irule{\hat{\Sigma}{_\ell}({x}) \subseteq \vartheta \subseteq  \Theta(\ell)} {\nEFORM{x}{\vartheta} }
\\[4ex] 
\irule{\begin{array}{c} \bigwedge_{i=1}^k\,
\nEFORM{E_i}{\vartheta_i} \; \wedge \\[.5ex] 
\forall\, \hat{v}_1,..,\hat{v}_k: 
\bigwedge_{i=1}^k\, \hat{v}_i \in \vartheta_i\
\\
 \Rightarrow
\begin{array}{l}
\  \lfloor \{\hat{v}_1,..,\hat{v}_k\}_{k_0}^{\ell} \rfloor_{d}  \in \vartheta 
\\
\ \wedge \ \vartheta \subseteq  \Theta(\ell)
\end{array}
          \end{array}}
         {\nEFORM{\{E_1,..,E_k\}_{k_0}}{\vartheta}}
\ \ \
\irule{\begin{array}{c} \bigwedge_{i=1}^k\,
\nEFORM{E_i}{\vartheta_i} \; \wedge \\[.5ex] 
\forall\, \hat{v}_1,..,\hat{v}_k: 
\bigwedge_{i=1}^k\, \hat{v}_i \in \vartheta_i\ 
\\
\Rightarrow
\begin{array}{l}
\ \lfloor f^\ell(\hat{v}_1,..,\hat{v}_k) \rfloor_{d} \in \vartheta 
\\
\ \wedge \ \vartheta \subseteq  \Theta(\ell)
\end{array}
          \end{array}}
         {\nEFORM{f(E_1,..,E_k)}{\vartheta}}
\end{array} 
\] 
\caption{Analysis of terms $\nEFORM{E}{\vartheta}$.}\label{analysisT}
\end{table*}
}
The result of our CFA is a quadruple $(\hat{\Sigma},\kappa,\Theta,\alpha)$
(a pair $(\hat{\Sigma}, \Theta)$ when analysing a term $E$, resp.),
called \emph{estimate} for $N$ (for $E$, resp.), that satisfies the
judgements defined by the axioms and rules of Tables~\ref{analysis} (and~\ref{analysisT}).
For this we introduce the following {\em abstract domains}:
\begin{itemize}
\item
{\it abstract store}
$
\hat{\Sigma} = \bigcup_{\ell \in \cal L}\,\hat{\Sigma}{_\ell}: {\cal X} \cup {\cal I}_\ell \ \rightarrow\  2^{\widehat{\cal V}}
$
where each {\em abstract local store} $\hat{\Sigma}_{\ell}$ approximates the concrete local store $\Sigma_{\ell}$, 
by associating with each location a set of abstract values that represent the possible concrete values that the location may stored at run time.
\item {\it abstract network environment} 
$\kappa : {\cal L} \ \rightarrow\  {\cal L} \times \bigcup_{i = 1}^k \widehat{\cal V}^i$ (with
$\widehat{\cal V}^{i+1} = \widehat{\cal V} \times \widehat{\cal V}^i$ and $k$ maximum arity of messages), that includes all the messages that may be received by the node 
$\ell$.
\item
{\em abstract data collection} $\Theta: {\cal L} \ \rightarrow\ 2^{\widehat{\cal V}}$
that, for each node $\ell$, approximates the set of values that the node computes.
\item
{\em abstract node action collection} $\alpha: {\cal L} \times {\cal J} \rightarrow Act$ that includes all the annotated actions $\gamma$ of the
actuator $j$ that may be triggered in the node $\ell$. We will write $\alpha_{_{\ell}}(j)$ in place of $\alpha(\ell,j)$.
\end{itemize}

\noindent
For each term $E$, the judgement $\nEFORM{E}{\vartheta}$, defined by the rules in \tablename~\ref{analysisT}, expresses that $\vartheta \in  \widehat{\cal V}$ is an acceptable estimate of the set of values that $E$ may evaluate to in $\hat{\Sigma}_\ell$.
A sensor identifier and a value evaluate to the set $\vartheta$, provided that their abstract representations belong to $\vartheta$.
Similarly a variable $x$ evaluates to $\vartheta$, if this includes the set of values bound to $x$ in $\hat{\Sigma}_{\ell}$.
The rule for analysing the encryption term produces the set $\vartheta$.
To do that
(i) for each term $E_i$, it finds the sets $\vartheta_i$, and
(ii) for all $k$-tuples of values $(\hat{v}_1,\cdots,\hat{v}_k)$ in $\vartheta_1\times\cdots\times\vartheta_k$, it checks if
the abstract values $\{\hat{v}_1,\cdots,\hat{v}_k\}_{k_0}^{\ell}$ belong to $\vartheta$.
The last rule analyses the application of a $k$-ary function $f$ to produce the set $\vartheta$.
To do that
(i) for each term $E_i$, it finds the sets $\vartheta_i$, and
(ii) for all $k$-tuples of values $(\hat{v}_1,\cdots,\hat{v}_k)$ in $\vartheta_1\times\cdots\times\vartheta_k$, it checks if
the abstract values $f^{\ell}(\hat{v}_1,\cdots,\hat{v}_k)$ belong to $\vartheta$.
Recall that the special abstract value $\top^{\ell}$ will end up in $\vartheta$ if the depth of the abstract functional term or the encryption term exceeds $d$, and it represents all the functional terms with nesting greater than $d$.
%

Moreover, in all the rules for terms, we require that $\Theta(\ell)$ includes all the abstract values included in $\vartheta$.
This guarantees that only those values actually used are tracked by $\Theta$, in particular those of sensors.

In the analysis of nodes we focus on which values can flow on the network and which can be assigned to variables.
The judgements have the form $\nNAFORM{N}$
and are defined by the rules in \tablename~\ref{analysis}. 
The rules for the \emph{inactive node} and for \emph{parallel composition} are standard.
Moreover, the rule for a single node $\ell:[B]$ requires that its component $B$ is analysed, with the further judgment
$\nPAFORM{\ell}{B}$, where $\ell$ is the label of the enclosing node.
The rule connecting actual stores $\Sigma$ with abstract ones $\hat{\Sigma}$ requires the locations of sensors to contain the corresponding abstract values.
The rule for sensors is  trivial, because we are only interested in who will use their values, and that for actuators is equally simple.
{\small
\begin{table*}[!bt]
 \[
 \begin{array}{c}
\irule{}
{\nNAFORM{\NIL}}
\qquad
\irule{\nPAFORM{\ell}{B}}
{\nNAFORM{\ell:[B]}}
\qquad
\irule{\nNAFORM{N_1} \wedge \nNFORM{N_2}}
      {\nNAFORM{N_1 \ | \ N_2}}
\\[5ex]
\irule{\forall\, i \in {\cal I_\ell}. \ i^\ell \in \hat{\Sigma}_\ell(i)}{\nPAFORM{\ell}{\Sigma}}
         \qquad\qquad
\irule{}{\nPAFORM{\ell}{S}}
          \qquad\qquad
\irule{}{\nPAFORM{\ell}{A}}
\\[4ex]
\irule{\begin{array}{c}
          \bigwedge_{i=1}^{k}\; \nEFORM{E_i}{\vartheta_i} \  \wedge \ 
          \nPAFORM{\ell}{P} \ \wedge
          \\[.5ex]
          \forall \hat{v}_1,\cdots,\hat{v}_k:\; \bigwedge_{i=1}^k\, \hat{v}_i \in \vartheta_i\
          \Rightarrow
          \forall \ell' \in L:
         (\ell,  \mess{\hat{v}_1,\cdots,\hat{v}_k}) \in  \kappa(\ell')           \end{array}}
         {\nPAFORM{\ell}{\OUTM{E_1,\cdots,E_k}{L}.\,P}}
\\[4ex]
\irule{
\begin{array}{c}
              \bigwedge_{i=1}^{j}\;  \nEFORM{E_i}{\vartheta_i}  \ \wedge \\[.5ex]
          \forall (\ell', \mess{ \hat{v} _1,\cdots,\hat{v}_k}) \in \kappa(\ell):\;Comp(\ell',\ell) 
\Rightarrow
\left(          \bigwedge_{i=j+1}^{k}\;  \hat{v}_i \in \hat{\Sigma}{_\ell}({x_i})\ \wedge \
           \nPAFORM{\ell}{P} 
\right)
\end{array}
          }
          {\nPAFORM{\ell}{\INPS{E_1,\cdots,E_j}{x_{j+1},\cdots,x_k}{P}}}
\\[6ex]
\irule{
\begin{array}{c}
             \nEFORM{E}{\vartheta} \  \wedge \ \bigwedge_{i=1}^{j}\;  \nEFORM{E_i}{\vartheta_i} \  \wedge \\[.5ex]
           \{\hat{v}_1, \cdots, \hat{v}_k \}^{\ell}_{k_0} \in \vartheta  \Rightarrow
\left(
           \bigwedge_{i=j+1}^{k}\;  \hat{v}_i \in \hat{\Sigma}{_\ell}({x_i})\ \wedge \
           \nPAFORM{\ell}{P} 
\right)
\end{array}
          }
          {\nPAFORM{\ell}{{\DECSO{E}{E_1,\cdots,E_j}{x_{j+1},\cdots,x_k}{k_0}{}{P}}}}
\\[5ex]
\begin{array}{ll}
\irule{
\begin{array}{c} 
\nEFORM{E}{\vartheta} \  \wedge \\ \forall \, \hat{v} \in \vartheta \  \Rightarrow
 \hat{v}  \in \hat{\Sigma}_{_\ell}(x) \ \wedge
\nPAFORM{\ell}{P}
\end{array}}
      {\nPAFORM{\ell}{x : = E}.\,{P}}
\qquad
&
\irule{{\gamma \in \alpha_\ell(j) \wedge \nPAFORM{\ell}{P}}}
      {\nPAFORM{\ell}{ \OUTS{j,\gamma}{P}
}}
\end{array}

\\[6ex]
\irule{\begin{array}{c} 
\nEFORM{E}{\vartheta} \; \wedge 
\\
           \nPAFORM{\ell}{P_1} \; \wedge \
           \nPAFORM{\ell}{P_2}
           \end{array}
         }
          {\nPAFORM{\ell}{E?P_1 : P_2}}
\quad
\irule{\nPAFORM{\ell}{B_1} \wedge \nPAFORM{\ell}{B_2}}{\nPAFORM{\ell}{B_1 \| \ B_2}}
\\[4ex]
\irule{}{\nPAFORM{\ell}{\NIL}}
\qquad\qquad
\irule{\nPAFORM{\ell}{\lfloor \mu h. \ P \rfloor_d}}{{\nPAFORM{\ell}{\mu h. \ P}}}
\qquad\qquad
\irule{}
{\nPAFORM{\ell}{h}}
\end{array} 
\] 
\caption{Analysis of nodes $\nNAFORM{N}$, and of node components
$\nPAFORM{\ell}{B}$.}\label{analysis}
\end{table*}
}

The axioms and rules for processes (in \tablename~\ref{analysis}) require that an estimate is also valid for the immediate sub-processes.
The rule for $k$-ary \emph{multi-output}
(i)  finds the sets $\vartheta_i$, for each term $E_i$; and 
(ii) for all $k$-tuples of values $(\hat{v}_1,\cdots,\hat{v}_k)$ in $\vartheta_1\times\cdots\times\vartheta_k$, it checks if
they belong to $\kappa(\ell' \in L)$, i.e.\ they can be received by the nodes with labels in $L$.
In the rule for \emph{input} the terms $E_1, \cdots, E_j$ are 
used for matching values sent on the network. 
Thus, this rule checks whether 
(i) these first $j$ terms have acceptable estimates $\vartheta_{i}$; 
(ii) and for any
message $(\ell', \mess{\hat{v}_1,\cdots,\hat{v}_j,\hat{v}_{j+1},\ldots,\hat{v}_k})$ in $\kappa(\ell)$ 
(i.e.~in any message predicted to be receivable by the node $\ell$), such that the two nodes can communicate ($Comp(\ell',\ell)$),
the values $\hat{v}_{j+1},\ldots,\hat{v}_k$
are included in the estimates for the variables $x_{j+1},\cdots,x_k$.
The rule for {\em decryption} is similar: it also requires that the keys coincide.
The rule for {\em assignment}
requires that all the values $\hat{v}$ in $\vartheta$, the estimate for $E$, belong to $\hat{\Sigma}_{_\ell}(x)$.
The next rule predicts in the component $\alpha$ that a process at node $\ell$ may trigger the action $\gamma$ for the actuator $j$.
The rule for $\mu h.\, P$ reflects our choice of limiting the depth of function applications: the iterative process is unfolded $d$ times, represented by $\lfloor \mu h. \ P \rfloor_d$.
The remaining rules are as expected.


Consider the example in Sect.~\ref{sec:example} and the process 
$P_{cp} = \mu h. (z := 1).(z' := noiseRed(z)).\OUTM{z'}{\{\ell_a\}}. \ h$.
Every valid CFA estimate must include at least the following entries (assuming the depth $d = 4$):
\[
\begin{array}{l}
(a)\, \hat{\Sigma}_{\ell_{cp}}(z) \supseteq \{1^{\ell_{cp}} \}
\qquad\qquad\qquad\qquad\ \ 
(b)\, \hat{\Sigma}_{\ell_{cp}}(z') \supseteq \{noiseRed^{^{\ell_{cp}}}(1^{\ell_{cp}}),1^{\ell_{cp}} \}
\\
(c)\, \Theta(\ell_{cp}) \supseteq \{1^{\ell_{cp}},noiseRed^{^{\ell_{cp}}}(1^{\ell_{cp}}) \} 
\quad
(d)\, \kappa(\ell_a) \supseteq \{ (\ell_{cp},\mess{noiseRed^{^{\ell_{cp}}}(1^{\ell_{cp}})})\}
\end{array}
\]
All the following checks must succeed:
$\nPAFORM{\ell_{cp}}{\mu h. (z := 1).(z' := noiseRed(z)).\OUTM{z'}{\{\ell_a\}}.h}$ because
 $\nPAFORM{\ell_{cp}}{ (z := 1).(z' := noiseRed(z)).\OUTM{z'}{\{\ell_a\}}}$, that in turn holds
because (i) $1^{\ell_{cp}}$ is in $\hat{\Sigma}_{\ell_{cp}}(z)$ by ({\it a}) ($\nEFORM{1}{\vartheta} \ni 1^{\ell_{cp}}$); and because 
(ii) $\nPAFORM{\ell_{cp}}{ (z' := noiseRed(z)).\OUTM{z'}{\{\ell_a\}}}$, that holds;
because (i) $noiseRed^{^{\ell_{cp}}}(1^{\ell_{cp}})$ is in $\hat{\Sigma}_{\ell_{cp}}(z')$ by ({\it b}) since
  $\nEFORMM{\ell_{cp}}{noiseRed(z)}{\vartheta}$ with $noiseRed^{^{\ell_{cp}}}(1^{\ell_{cp}}) \in \vartheta$; and because
(ii) $\nPAFORM{\ell_{cp}}{\OUTM{z'}{\{\ell_a\}}}$ that holds because $\kappa(\ell_a)$ includes
$(\ell_{cp},\mess{noiseRed^{^{\ell_{cp}}}(1^{\ell_{cp}})})$ by ({\it d}).
By considering instead the process $P_{p,1}$ below we have that $\alpha_{\ell_p}(j)$ includes $\mathsf{turnon},\mathsf{turnoff}$.
\begin{align*} 
P_{p,1} = \mu h. & (x_1:= 1.\, x_2:= 2.\, x_3:= 3.\, x_4:= 4). \\
& (x_4 = true)\ ?\  \\
& \qquad (x_1 \leq th_1 \land x_2 \leq th_2)\ ?\ \\
 &  \qquad \qquad \qquad \qquad\ \  \ (x_3 > th_3)\ ?\ \OUTS{5,\mathsf{turnon}}{\OUTM{x_4}{L_p}}.\ h\\
 &  \qquad \qquad \qquad \qquad \qquad \qquad \quad \ \ : \ \OUTM{\mathsf{err}, \ell_p}{\{\ell_s\}}. \ h \\
  &  \qquad \qquad \qquad \qquad \qquad \qquad : h \\
    & \ \ \ \ \qquad \qquad: \  \OUTS{5,\mathsf{turnoff}}{h}
\end{align*}
\paragraph{Correctness of the analysis.}

Our analysis respects the operational semantics of \IoTLySa.
As usual, we can prove a subject reduction result for our analysis and the existence of a (minimal) estimate.
The proofs benefit from an instrumented denotational semantics for expressions, the values of which are pairs $\langle v, \hat{v}\rangle$.
Consequently, the store ($\Sigma^i_\ell$ with a $\bot$ value) and its update are accordingly extended.

Just to give an intuition, we will have $\dsem{v}_{\Sigma^i_\ell}^i  =  (v, v^\ell)$, and the assignment $x : = E$ will result in the updated store $\Sigma^i_\ell \{(v, v^\ell)/x\} $, where $E$ evaluates to $(v, v^\ell)$.
Clearly, the semantics used in \tablename~\ref{opsem} is $\dsem{v}^i _{\downarrow_1}$, the projection on the first component of the instrumented one.
Back to the example of Sect.~\ref{sec:example}, the assignment $z' := noiseRed(z)$ of the process $P_{cp}$ stores the pair $(v, noiseRed^{\ell_{cp}}(1^{\ell_{cp}}))$, where the first component is the actual value resulting from the function application, and the second is its abstract counterpart (note that the sensor value is abstracted as $1^{\ell_{cp}}$).


Since the analysis only considers the second component of the extended store, it is immediate defining when the concrete and the abstract stores agree:
$\Sigma_\ell^i \bowtie \hat{\Sigma}_{\ell}$ if and only if
$w \in \cal{X} \cup \cal{I} _\ell$ such that $\Sigma^i_\ell (w) \neq \bot$
implies
$(\Sigma^i_\ell (w))_{\downarrow_2} \in  \hat{\Sigma}_\ell(w)$.

The following theorems (whose proofs follow the usual schema), which correspond to the ones in~\cite{BDFG_Coord16}, establish the correctness of our CFA.
Also, we can prove the existence of a minimal estimate, which depends on the fact that the set of
analysis estimates can be partially ordered and constitutes a Moore family.

\begin{restatable}[Subject reduction]{theorem}{subjectreduction}\label{SRtheorem}
\  \\
If $\nNAFORM{N}$ and $N \rightarrow N'$ and $\,\forall\Sigma_\ell^i$ in $N$ it is $\Sigma_\ell^i \bowtie \hat{\Sigma}_{\ell}$, 
then $\nNAFORM{N'}$ and $\,\forall \Sigma{_\ell^i}'$ in $N'$ it is $\Sigma{_\ell^i}' \bowtie \hat{\Sigma_\ell}$.
\end{restatable}

\noindent
The following corollary of subject reduction shows that our CFA 
guarantees that $\kappa$ predicts all the possible inter-node communications.

\begin{restatable}{corollary}{corTheta}\label{cor:Theta}\ \\
Let $N \xrightarrow{\mess{v_1,\dots,v_n}}_{\ell_{1}, \ell_{2}} N'$ denote a reduction in which the message sent by node $\ell_1$ is received by node $\ell_2$.
If $\nNAFORM{N}$ and $N \xrightarrow{\mess{v_1,\dots,v_n}}_{\ell_1,\ell_2} N'$ then
it holds $(\ell_{1},\mess{\hat{v}_1,\dots,\hat{v}_n}) \in \kappa(\ell_{2})$, 
where $\hat{v_i} = v_{i\downarrow_2}$.
\end{restatable}

%

\noindent
Back again to our example, we have that 
$1^{\ell_{cp}} \in \Theta(\ell_{cp})$, where $(\dsem{1}^1_{\Sigma^1_{\ell_{cp}}})_{\downarrow_2} = 1^{\ell_{cp}}$,  
and where $v$ is the actual value received by the first sensor.
Similarly, we have that
$(\ell_{cp},\mess{\hat{v}}) \in \kappa(\ell_{a})$, 
where $\hat{v} = v_{\downarrow_2}$.

\paragraph{Action Tracking}
The new component  of the analysis $\alpha$ allows us to perform other checks on actuators that 
might suggest to use a simpler actuator if some of its actions are never triggered, or even to remove it if it is never used.
In the following definitions  $\rightarrow^*$ denotes
the reflexive and transitive closure of $\rightarrow$.

\begin{definition}
The node $N$ with label $\ell$ {\em can never fire} an action $\gamma$ on actuator $j$ if whenever
$N \rightarrow^* N'$ then there is no transition
$N' \rightarrow N''$ obtained by applying  the rule \mbox{(A-com)} on the action $\gamma$ on actuator $j$.
\end{definition}
%
%
\begin{definition}
An action $\gamma$ is {\em never triggered} {by $N = \ell : [B]$} on $j$ if
$\nNAFORM{N}$, and $\gamma \not\in \alpha_{\ell}(j)$.
\end{definition}
Since we are over-approximating, the presence of $\gamma$ in $\alpha_{\ell}(j)$ only implies that the node {\em may} trigger the action $\gamma$ on actuator $j$.
In our running example, the inclusion
$\{\mathsf{turnon},\mathsf{turnoff}\} \subseteq \alpha_{_{\ell_{p}}}(5)$ only says that the two actions are reachable.
Instead if, e.g.~the output $\mess{5,\mathsf{turnoff}}$ were erroneously omitted in the last line of the specification of process $P_{p,1}$, then our CFA
would detect that $\mathsf{turnoff} \not\in \alpha_{_{\ell_{p}}}(5)$.

\begin{theorem}
Given a node $N$ with label $\ell$, if an action $\gamma$ is {\em never triggered} on $j$, then 
$N$ with label $\ell$ {\em can never fire} an action $\gamma$ on actuator $j$.
\end{theorem}
\begin{proof}
By contradiction suppose that there exists a transition $N' \rightarrow N''$, derived by using the rule  \mbox{(A-com)}, with action $\gamma$ and actuator $j$. 
Then a control process in $N'$ must include the prefix $\langle{j,\gamma}\rangle$. 
Now, because of Theorem~\ref{SRtheorem}, we have that $\nNAFORM{N}$ implies $\nNAFORM{N'}$. 
As a consequence, the analysis for $N'$ must hold for the process $\OUTS{j,\gamma}{P}$ and therefore $\gamma \in \alpha_{\ell}(j)$: contradiction.
\end{proof}

Similarly, at run time the actuator $j$ never fires an action if $j$ is {\em never used}, defined as follows:

\begin{definition}
An actuator $j$ is {\em never used} {in $N = \ell : [B]$} if
$\nNAFORM{N}$, and $\alpha_{\ell}(j) = \emptyset$.
\end{definition}
%

A major feature of our approach is that the two properties above, as well as all those mentioned in the next section and many more, can be verified on the estimates given by the CFA, with no need of recomputing them.

\section{Security properties}
	

\paragraph{Preventing Leakage}
Our CFA can be exploited to check some security properties, along the lines of~\cite{BBDNN_InfComp,phdBod}, by using some ideas of~\cite{Ab97}.
For the moment, our analysis implicitly considers the presence of only {\em passive} attackers, able 
to eavesdrop all the messages in clear and to decrypt if in possession of the right key.
We can take into account active attackers, by modifying our CFA in the style of~\cite{BBDNN_JCS}.

A common approach to system security is identifying the sensitive content of information and detect possible disclosures.
Hence, we partition values
into security classes and prevent classified information from flowing in clear or to the wrong places.
In the following, we assume as given the sets ${\cal S}_{\ell}$ and ${\cal P}_{\ell}$ of {\em secret} and of \emph{public}
values for the node $N$ with label ${\ell}$. 
It is immediate to have a hierarchy of
classification levels associated with values, instead of just two levels.

The designer classifies constants and values of sensors as secret or public. 
The computed values, both concrete and abstract, are partitioned through the concrete and the abstract operators $D_{cls}$ and $S_{cls}$.
The intuition behind them is that
a single ``drop'' of $\s$ turns to \emph{secret} the kind of a(n abstract) term.
Of course there is the exception for encrypted data:
what is encrypted is $\p$, even if it
contains secret components.
Accordingly, in the analysis we replace $\top$ with the two abstract terms $\top_{\!\!s}$ and  $\top_{\!\!p}$ for abstracting secret and public terms, respectively.


\begin{definition}
Given ${\cal S}_{\ell}$ and ${\cal P}_{\ell}$, the {\em static} operator
$S_{cls}: \widehat{\cal V}  \rightarrow \{\s, \p\}$
is defined as follows:
\[
\begin{array}{ll}
S_{cls}(\top_{\!\!s}) = \s
 &
 S_{cls}(\top_{\!\!p}) = \p
 \\
S_{cls}(i^{\ell}
) = \casesdef{\!\!\! \s}{\mbox{ if }i^{\ell}  \in {\cal S}_\ell}{\!\!\! \p}{\mbox{ if }i^{\ell}  \in
{\cal P}_\ell}
&
S_{cls}(v^{\ell} ) = \casesdef{\!\!\! \s}{\mbox{ if }v^{\ell}  \in {\cal S}_\ell}{\!\!\! \p}{\mbox{ if }v^{\ell}  \in
{\cal P}_\ell}
\\
S_{cls}(f^{\ell}(\hat{v}_1,\ldots,\hat{v}_k)) = S_{cls}(\{\hat{v}_1,\ldots,\hat{v}_k\})
&
S_{cls}(\{\hat{v}_1,\ldots,\hat{v}_k\}^\ell_{k_0}) = \p
\\
S_{cls}(W) = \casesdef{\s}
{\mbox{if } \exists \hat{v} \in W: S_{cls}(\hat{v}) = \s}
{\p}{\mbox{otherwise.}}
\end{array}
\]
\end{definition}
\begin{definition}
Given ${\cal S}_{\ell}$ and ${\cal P}_{\ell}$, the {\em dynamic} operator
$D_{cls}: {\cal V}  \rightarrow \{\s, \p\}$
is defined as:
\[
\begin{array}{ll}
D_{cls}(i) = \casesdef{\!\!\! \s}{\mbox{ if } i  \in {\cal S}_{\ell}}{\!\!\! \p}{\mbox{ if }i \in
{\cal P}_{\ell}}
&
D_{cls}(v) = \casesdef{\!\!\! \s}{\mbox{ if } v  \in {\cal S}_{\ell}}{\!\!\! \p}{\mbox{ if }v  \in
{\cal P}_{\ell}}
\\
D_{cls}(f(v_1,\ldots,v_k)) = 
D_{cls}(\{v_1,\ldots,v_k\})
&
D_{cls}( \{ v_1,\ldots,v_k\}_{k_0}) = \p
\\
D_{cls}(V) = \casesdef{\s}
{\mbox{if }\exists v \in V: D_{cls}(v) = \s}
{\p}{\mbox{otherwise.}}
\end{array}
\]
\end{definition}

Since our analysis computes information on the values exchanged during the communication, we can statically check whether a value, devised to be secret to a node $N$, is never sent
to another node.
We first give a dynamic characterisation of when a node $N$ never discloses its secret values, i.e.\ when neither it nor any of its derivatives can send a message that includes a secret value.
\begin{definition}
The node $N$ with label $\ell$ has {\em no leaks} w.r.t.~${\cal S}_{\ell}$ if $N \rightarrow^* N'$ and
there is no transition
$N' \xrightarrow{\mess{v_1,\dots,v_n}}_{\ell, \ell'} N''$
such that $D_{cls}(v_i) = \s$ for some $i$.
\end{definition}
The component $\kappa$ allows us to define a static notion that corresponds to
the above dynamic property: a node $\ell$ is \emph{confined} when all the values of its messages are public.
Confinement suffices to guarantee that $N$ has no leaks, as an immediate consequence of Corollary~\ref{cor:Theta}, where also the partition of values is considered. 

\begin{definition}
A node $N$ with label $\ell$ is {\em confined} w.r.t.~${\cal S}_{\ell}$ if 
\begin{itemize}
\item
$\nNAFORM{N}$ and
\item
$\forall \ell' \in {\cal L}$ such that $(\ell,  \mess{\hat{v}_1,\cdots,\hat{v}_k}) \in  \kappa(\ell')$ we have that
$\forall i. \ S_{cls}(\hat{v}_i) = \p$
 \end{itemize}
\end{definition}

\begin{theorem}
Given a node $N$ with label $\ell$, if $N$ is confined w.r.t.~${\cal S}_{\ell}$, then 
$N$ has no leaks w.r.t.~${\cal S}_{\ell}$.
\end{theorem}
\begin{proof}
By contradiction suppose that $N$ is confined and that $N \rightarrow N'$ and that there exists a transition $N' \xrightarrow{\mess{v_1,\dots,v_n}}_{\ell, \ell'} N''$
such that $\exists i : D_{cls}(v_i) = \s$. 
Then, by Corollary~\ref{cor:Theta},
it holds $(\ell,\mess{\hat{v}_1,\dots,\hat{v}_n}) \in \kappa(\ell')$, where $\hat{v}_i = v_{i\downarrow_2}$.
Now, since $N$ is confined, we know that $\forall i. \ S_{cls}(\hat{v}_i) = \p$: contradiction.
\end{proof}

Back to our running example, since the pictures of cars are sensitive data, one would like to check whether they are kept secret.
By classifying $1^{\ell_{cp}}$ as one of the secret elements for the node $\ell_{cp}$, we accordingly get that $noiseRed^{^{\ell_{cp}}}(1^{\ell_{cp}})$ is $\s$.
By
inspecting $\kappa$ 
we discover e.g.\  that the sensitive data of cars is sent in clear to the street supervisor $\ell_a$, so possibly violating privacy; indeed $\kappa(\ell_a) \supseteq \{ (\ell_{cp},\mess{noiseRed^{^{\ell_{cp}}}(1^{\ell_{cp}})})\}$.
To prevent this disclosure the pictures are better sent encrypted.

\paragraph{Communication policies}
Another way of enforcing security is by defining policies on communications that rule information flows among nodes, allowing some flows and 
forbidding others.
Below we consider a {\em no read-up/no write-down}  policy.
It is based on a hierarchy of clearance levels for nodes~\cite{BL73,Den82}, and
it requires that a node classified at a
high level cannot write any value to a node at a lower level,
while the converse is allowed; symmetrically a node at low level
cannot read data from one of a higher level.
%

For us, it suffices to classify the node labels with an assignment function $level : {\mathcal L} \rightarrow {\bf L}$, from the set of node lables to a given set of levels $\bf L$.
We then introduce a condition for characterising the allowed and forbidden flows.
For the policy above, the condition amounts to requiring that
$
level(\ell) \leq level(\ell')
$.
\begin{definition}
Given an assignment \emph{levels}, the node $N$ with label $\ell$ {\em respects the levels} if $N \rightarrow^* N'$ and
there is no transition 
\mbox{s.t.  
$N' \xrightarrow{\mess{v_1,\dots,v_n}}_{\ell, \ell'} N''$}
and $level(\ell) > level(\ell')$.

\end{definition}
Similarly, we define the following static notion that checks when flows
from a node $\ell$ to a node $\ell'$ are allowed.
Also in this case the static condition implies the dynamic one, as an immediate consequence of Corollary~\ref{cor:Theta}.
\begin{definition}
Given an assignment \emph{levels}, the node $N$ with label $\ell$ {\em safely communicates} if 
\begin{itemize}
\item
$\nNAFORM{N}$ and
\item
$\forall \ell' \in {\cal L}$ such that $(\ell,  \mess{\hat{v}_1,\cdots,\hat{v}_k}) \in  \kappa(\ell')$ we have that
$level(\ell) \leq level(\ell')$
 \end{itemize}
\end{definition}
%


\begin{theorem}
Given a node $N$ with label $\ell$, if $N$ {\em respects the levels} ${\bf L}$, then 
$N$ {\em safely communicates}.
\end{theorem}
\begin{proof}
By contradiction suppose that $N$ safely communicates and that $N \rightarrow N'$ and that there exists a transition $N' \xrightarrow{\mess{v_1,\dots,v_n}}_{\ell, \ell'} N''$
such that $level(\ell) > level(\ell')$. 
Then, by Corollary~\ref{cor:Theta},
it holds $(\ell,\mess{\hat{v}_1,\dots,\hat{v}_n}) \in \kappa(\ell')$, where $\hat{v}_i = v_{i\downarrow_2}$.
Now, since $N$ safely communicates, $level(\ell) \leq level(\ell')$: contradiction.
\end{proof}

Of course, one can mix the above checks to verify a composition of the considered properties, e.g.~by checking whether a particular secret value does not flow to a specific node, even if it has a higher level.

More in general, we can constrain communication flows according to a specific policy, by replacing the 
condition on the levels used above with a suitable predicate 
$
\phi(\ell,\ell')
$.
In our running example, we can check (the trivial fact) that the communication from node $N_a$ to $N_s$ is allowed by the policy, while those in the other direction are not. 
Also, note that a pretty similar approach can be used to deal with trust among nodes: instead of assigning nodes a security level, just give them a discrete measure of their trust.
Similarly, we could check whether a node is allowed to use a particular aggregation function.

\label{sec:extendedCFA}
	
\section{Conclusions}\label{sec:conclusion}

In the companion paper~\cite{BDFG_Coord16} we 
introduced the process calculus \IoTLySa\ for describing IoT systems, which has primitive constructs for describing the activity of sensors and of actuators, and suitable primitives for managing the coordination and communication capabilities of  smart objects.
This calculus is endowed with a Control Flow Analysis that statically predicts the interactions among nodes, how data spread from sensors to the network, and how data are put together.

Here, we extended IoT-LySa with cryptographic primitives and we enriched the static analysis in order to predict which actions actuators can perform.
Building over the analysis, we began an investigation about security issues, by checking some properties. 
The first one is secrecy: we check whether a value, devised to be secret to a node, is never sent to another node in clear.
The second property has to do with a classification of clearance levels of confidentiality: exploiting the analysis estimates we verify if  the predicted communications never flow from a high level node to a lower level one.

As future work, we would like to investigate
how sensors data affect actuators, by tracking the
information flow between sensors and actuators.
We also would like  to capture dependencies among the data sent by sensors and
the actions carried out by actuators.

To the best of our knowledge, 
only a limited number of papers addressed the specification and verification of IoT systems from a process algebraic perspective,
e.g.\ \cite{lanese13,merro16} to cite only a few.
Also relevant are the papers that formalised wireless, sensor and ad hoc networks, among which~\cite{lanese10,singh10,NNN10}.
We shared many design choices with the above-mentioned proposals, while the main difference consists in the underlying coordination model.
Ours is based on a shared store \`a la Linda instead of a message-based communication \`a la $\pi$-calculus.
Furthermore, differently from~\cite{lanese13,merro16}, 
our main aim is that of developing a design framework that includes a static semantics to support verification techniques and tools for certifying properties of IoT applications.

\bibliographystyle{eptcs}
\bibliography{biblio}



\end{document}